\documentclass[journal, a4paper, 11pt, onecolumn]{IEEEtran}

\usepackage[utf8]{inputenc}
\usepackage{graphicx, subcaption}
\usepackage{amsmath,bm,amssymb,mathabx, amsthm}
\usepackage{color}
\usepackage{cite,url}
\usepackage{lipsum}

\newtheorem{theorem}{Theorem}
\newtheorem{definition}{Definition}
\newtheorem{corollary}[theorem]{Corollary}

\newtheorem{remark}{Remark}

\newtheoremstyle{alg}
                 {10pt}
                 {0pt}
                 {}
                 {}
                 {\bfseries}
                 {}
                 {\newline}
                 {\thmname{#1} \thmnumber{#2}. \thmnote{\textit{\textbf{#3}}}}
                 
\theoremstyle{alg}
\newtheorem{algorithm}{Algorithm}

\title{Coded FFT and Its Communication Overhead}
\date{December 2017}
\author{Haewon Jeong, Tze Meng Low, and Pulkit Grover\\haewon@cmu.edu, lowt@andrew.cmu.edu, pulkit@cmu.edu}

\begin{document}
\maketitle

\begin{abstract}
We propose a coded computing strategy and examine communication costs of coded computing algorithms to make distributed Fast Fourier Transform (FFT) resilient to errors during the computation. We apply maximum distance separable (MDS) codes to a widely used ``Transpose'' algorithm for parallel FFT. In the uncoded distributed FFT algorithm, the most expensive step is a single ``all-to-all'' communication. We compare this with communication overhead of coding in our coded FFT algorithm. We show that by using a systematic MDS code, the communication overhead of coding is negligible in comparison with the communication costs inherent in an uncoded FFT implementation if the number of parity nodes is at most $o(\log K)$, where $K$ is the number of systematic nodes. 
\end{abstract}

\section{Introduction}
The FFT algorithm is the backbone of scientific computations which is used in number of applications, such as solving Poisson's equations for particle simulations \cite{ryne2011fft,hockney1965}. For solving differential equations with high accuracy, often, FFTs of a very large size are computed. Thus, practitioners implement the FFT over massively parallel processors to speed up scientific simulations \cite{FFTW05,pippig2013pfft, duy2014decomposition}. 

As we head into the era of \textit{exascale} computing, with $>$ 100,000 processing nodes, we expect to see more frequent faults in the course of computation~\cite{bergman2008exascale}.  Traditionally, a   ``checkpointing'' technique is used to mitigate faults, where the state of the computation is stored offline at regular intervals and the computation restarts at the most recent checkpoint when an error is detected.
However, error rate of large parallel systems is projected to reach a point where present checkpoint/restart methods will no longer be viable ~\cite{ferreira2011replication,schroeder2010large}. We thus need a more efficient way to mitigate faults in large-scale numerical algorithms by exploiting algorithm-specific features.

In this work, we propose a coded computing approach to make large-scale
parallel FFT more fault tolerant. Coded computing combines computing and error correcting codes (ECCs) to mitigate unreliable processors in distributed computing \cite{huang2012codes, lee2017speeding,tandon2016gradient,dutta2016short,ferdinand2016anytime,li2016unified,reisizadeh2017coded,yu2017polynomial,yang2017coded,dutta2017coded}. 
Applying coding to FFT algorithm poses unique challenges because distributed FFT computation has an unavoidable \emph{``all-to-all communication''} step where all the processors have to exchange data with each other in the middle of the computation. Fault-tolerance measure applied at the end of the computation does not provide good enough resilience as any faults occurring prior to the all-to-all communication step can block the entire computation if no proper fault recovery is employed before the communication stage. This requires decoding and encoding in the middle of the computation. Since data is already distributed over many processors at this point, encoding or decoding must be done across distributed nodes, which incurs \emph{a communication overhead}.

It is well known that in distributed computing, communication is the bottleneck, not the computation. In this paper, we focus on identifying communication overhead of coding and constructing a coding scheme whose communication overhead is negligible compared to the built-in communication cost of the FFT algorithm. By using a linear model \cite{fraigniaud1994methods} to estimate communication cost, we show that as long as the number of parity servers is at most $o(\log K)$, where $K$ is the number of systematic processors, communication overhead of our proposed coded FFT algorithm can be amortized. 

Using ECCs for fault-tolerant FFT computation has been extensively studied in  algorithm-based fault-tolerance (ABFT) literature. The ABFT philosophy, first suggested by Huang and Abraham in 1984 \cite{huang1984abft}, is adding checksums in the beginning of the computation and comparing it against the checksum of outputs at the end of the computation in order to detect errors that happened during the computation. In essence, it uses the simplest error-detection code for computations. The ABFT technique was first applied to FFT by Jou and Abraham in 1986 \cite{jou1986} and since then substantial research has been done to improve on hardware/time overhead by designing new weighted checksums that can be efficiently implemented on an FFT circuit \cite{Jou88,Tao93, Wang94,Oh95}. A recent work \cite{yu2017coded} has extended these work to MDS codes. None of these works, however, have considered communication overhead of their algorithms, because most of them were focused on circuit-level fault tolerance where communication is not expensive.

Our contributions in this paper can be summarized as follow:
\begin{itemize}
\item We propose the first coding approach applied to the widely-used \emph{``transpose''} algorithm for distributed FFT computation~\cite{FFTW05}.
\item We show lower and upper bounds for new types of communication, \emph{``mutil-reduce''} and \emph{``multi-broadcast''} (Section~\ref{sec:coding}).
\item We show that as long as the number of parity nodes is at most $o(\log K)$, communication overhead of coding can be negligible compared to the communication cost of the FFT algorithm itself (Section~\ref{sec:comm_cost}).
\end{itemize}


\section{System Model and Preliminaries} \label{sec:model}

\subsection{Distributed Computing Model} \label{subsec:comp_model}
We assume that we have total of $P$  processors which have the same computational capabilities and memory. Among $P$ processors, $K$ of them are \emph{``systematic processors''} which store the original data and the remaining  $P-K$ processors are \emph{``parity processors''} that store encoded parity symbols. We assume that  $P-K \ll K$. 

We assume \emph{fully-distributed} setting where no central processor is present during the computation and the processors do not have any shared memory. Data located at different processors can be shared only through explicit communication between two processors. We will use \emph{``processor''} and \emph{``node''} interchangeably in this paper. 

Using $P$ processors, we want to compute $N$-point FFT: 
\begin{equation} \label{eq:N_fft}
\bm{Z} = F_N \bm{x}
\end{equation}
where $\bm{x}$ is a length-$N$ input data vector, $F_N$ is an $N$-by-$N$ DFT matrix ($\omega_N$: $N$-th root of unity) represented as
\begin{equation}
F_N = \small{\begin{bmatrix}
\omega_N^0 & \omega_N^0 & \cdots & \omega_N^0 \\
\omega_N^0 & \omega_N^1 & \cdots & \omega_N^{N-1} \\
\vdots & \vdots & \ddots & \vdots \\
\omega_N^0 & \omega_N^{N-1} & \cdots & \omega_N^{(N-1)^2}
\end{bmatrix}},
\end{equation}
and $\bm{Z}$ is a length-$N$ vector of the Fourier transform of $\bm{x}$. We assume that $N$ is very large so that the data cannot be stored in one processor. In the beginning, each processor has a segment of consecutive values of the input vector $\bm{x}$, e.g., Processor 1 has $\begin{bmatrix}
x_1 & x_2 & \cdots & x_{N/K}
\end{bmatrix}^T$. We add one more assumption here: $K = o(\sqrt{N})$. This is to ensure that each node can receive at least one full row or column once we arrange $N$-length vector to an $N_1$-by-$N_2$ matrix (explained in \label{subesc:fft_alg}) where $N_1$ and $N_2$ are $\Theta(\sqrt{N})$. From this assumption, the storage size of each node has to grow as $N$ grows.



\subsection{Fault Model} \label{subsec:fault_model}
In this work, we consider an erasure model of faults where we lose the entire output of a failed node. A node failure can happen when a node dies after a random fault,  or it is a straggler that is unable to finish its computation.

\subsection{Communication Model} \label{subsec:comm_model}
We assume a fully-connected network of processors where any processor can send and receive data from every other processor directly.\footnote{In many high-performance computing (HPC) settings, nodes are connected in a fixed topology, such as hypercubes, but we believe that our result on a fully-connected network is a good starting point.}  We also assume that a processor has one duplex port, which means that at a given time, a processor can send data to one processor and receive data from another processor simultaneously.\footnote{We believe that this can be easily extended to $k$-port model where each node has $k$ duplex ports.} For example, Processor 1 can send data to Processor 2 and receive data from Processor 3 at the same time, but it cannot send data to Processor 2 and 3 at the same time.

We use the basic $\alpha$-$\beta$ model to estimate the point-to-point communication cost. In the $\alpha$-$\beta$ model, the time to send or receive a message of $s$ bytes is : 
\begin{equation}
T = \alpha  + s \cdot \beta 
\end{equation}
Here, $\alpha$ is startup time to establish a connection between two nodes, and $\beta$ is the bandwidth cost required to transfer one symbol. For an algorithm that requires multiple rounds of message exchanges, total communication time can be written as follows:
\begin{equation}
T = C_1 \alpha + C_2 \beta,
\end{equation}
where $C_1$ is the number of communication rounds, $C_2$ is the number of symbols communicated in a sequence. To be more precise, if we denote $b_i$ as the maximum number of symbols communicated between two nodes at the $i$-th round, $C_2$ can be written as: 
\begin{equation}
C_2 =  \sum_{i=1}^{C_1} b_i.
\end{equation}
This is because the next round does not start until the previous round is completed, and the bandwidth latency for each round is dominated by the largest message.
Symbols can have different units such as bits or bytes, but in this work we do not specify any units.

\subsection{Distributed FFT algorithm} \label{subsec:fft_alg}
We want to explain the \emph{``Transpose''} algorithm that is commonly used in high-performance FFT libraries \cite{FFTW05}. It uses the Cooley-Tukey technique to break down $N$-point FFT into smaller FFTs of size $N_1$ and $N_2$ where $N= N_1 N_2$. (\ref{eq:N_fft}) can be rewritten as
\begin{align*}
Z_{k} &= \sum_{n = 0}^{N-1} \omega_{N}^{nk} x_n \\
&= \sum_{n_1 = 0}^{N_1-1} \omega_{N_1}^{n_1 k_1} t_{n_1,k_2} \sum_{n_2=0}^{N_2-1} \omega_{N_2}^{n_2 k_2} x_{n_2 N_1 + n_1} \label{eq:1D_FFT_twiddle}
\end{align*}
where $k = k_1 N_2 + k_2, \; k_1 = 0, \cdots N_1-1, \; \textnormal{and } k_2 = 0, \cdots, N_2-1$. The terms $t_{n_1,k_2}$;s are called twiddle factor which are equal to $\omega_N^{k_2 n_1}$. 

We can now compute $N$-point FFTs in two steps. In the first step, each processor is assigned to compute $N_1/K$ FFTs of length $N_2$. Then the processors transpose the data (requiring communication) and compute $N_2/K$ FFTs of size $N_1$ in the second step. 
Between the first and the second step, we have to multiply twiddle factors. This complicates our coding approach since multiplying twiddle factors is an element-wise multiplication of two matrices (Hadamard product) which does not commute with matrix-matrix multiplication (See Remark~\ref{rem:enc_adv}). We now explain the algorithm in  detail:
\begin{figure*}[ht]
\centering
\includegraphics[width=0.9\textwidth]{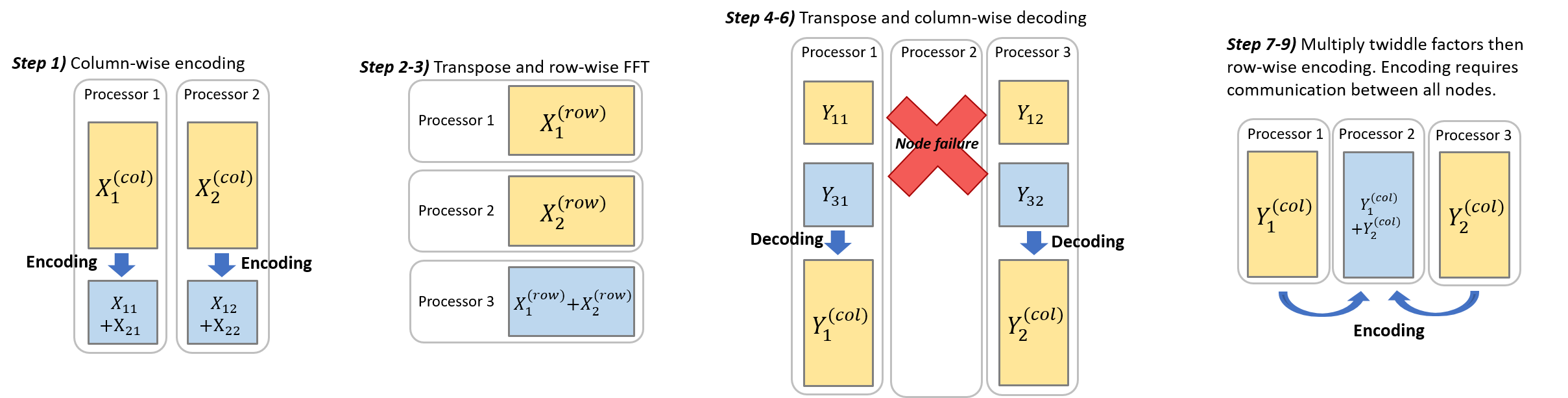}
\caption{This diagram summarizes encoding and decoding steps in Algorithm~\ref{alg:coded_fft} with an example of $P = 3, K = 2$.}
\end{figure*}
\begin{algorithm}[Uncoded Distributed FFT Algorithm (Transpose Algorithm)] \label{alg:uncoded_fft}
\mbox{}
\vspace*{-\parsep}
\vspace*{-\baselineskip}

\begin{enumerate}
\item \label{stp:uc_rarrange} Rearrange the input data $\bm{x}$ into $X$:
\begin{align*}
X = \small \begin{bmatrix}
x_1 & x_{N_1+1} & \cdots & x_{(N_2-1)N_1+1} \\
\vdots & \vdots & \ddots & \vdots \\
x_{N_1} & x_{2 N_1} & \cdots & x_{N_1 N_2}
\end{bmatrix} = \small \begin{bmatrix}
X_1^{(\textnormal{row})} \\
\vdots \\
X_K^{(\textnormal{row})}
\end{bmatrix} = \small \begin{bmatrix}
X_1^{(\textnormal{col})} & \cdots & X_K^{(\textnormal{col})}
\end{bmatrix}.
\end{align*}
We use $X_i^{\textnormal{(row)}}$'s ($X_i^{\textnormal{(col)}}$'s) to denote equal-sized submatrices of $X$ divided horizontally (vertically). From our system assumption, in the beginning, the $i$-th processor has $X_i^{\textnormal{(col)}}$ \footnote{This assumption is coming from that it is more natural for a processor to store contiguous data without the knowledge that the next computation is going to be FFT. If we assume that processors have row-wise data in the beginning, we can avoid the first transpose step. This does not change the result in Theorem~\ref{thm:main_result} in scaling sense.}. To begin the distributed FFT computation, we transpose the data distributed over $K$ processors so that the $i$-th processor can now have $X_i^{\textnormal{(row)}}$.

\item Compute $N_1/K$ row-wise FFTs of size $N_2$ at each processor.
\begin{equation*}
Y_i^{\textnormal{(row)}} = X_i^{\textnormal{(row)}} F_{N_2} 
\end{equation*}

\item \label{stp:uc_transpose} Transpose the data so that the $i$-th processor has $Y_i^{\textnormal{(col)}}$. 
\begin{equation*}
Y = \small{\begin{bmatrix}
Y_1^{\textnormal{(row)}} \\
\vdots \\
Y_K^{\textnormal{(row)}}
\end{bmatrix} = \begin{bmatrix}
Y_1^{\textnormal{(col)}} & \cdots & Y_K^{\textnormal{(col)}}
\end{bmatrix}} 
\end{equation*}

\item Multiply twiddle factors at each processor.
\begin{equation*}
Y_i^{\textnormal{(col)}} = T_{N,i}^{\textnormal{(col)}} \circ Y_i^{\textnormal{(col)}}
\end{equation*}
where $\circ$ represents Hadamard product and $T_N$ is a matrix of twiddle factors
\begin{align}
T_N = \small{\begin{bmatrix}
\omega_N^{0} & \omega_N^{0} & \cdots & \omega_N^{0} \\
\omega_N^{0} & \omega_N^{1} & \cdots & \omega_N^{N_2-1} \\
\vdots & \vdots & \ddots & \vdots \\
\omega_N^{0} & \omega_N^{N_1-1} & \cdots & \omega_N^{(N_1-1)(N_2-1)} \\
\end{bmatrix}} = \begin{bmatrix}
T_{N,1}^{\textnormal{(col)}} & \cdots & T_{N,K}^{\textnormal{(col)}} 
\end{bmatrix}. \nonumber
\end{align}

\item Compute $N_2/K$ column-wise FFTs of size $N_1$ at each processor.
\begin{equation}
Z_i^{\textnormal{(col)}} =  F_{N_1} Y_i^{\textnormal{(col)}}.
\end{equation}

\end{enumerate}

\end{algorithm}

\section{Coded Distributed FFT}\label{sec:coding}

We will now explain our coding strategy for the distributed FFT algorithm. In the \emph{uncoded distributed FFT algorithm} explained in  \ref{subsec:fft_alg}, if there is any failed node in the first step, we cannot proceed to the transpose step because every node requires pieces of output from all the other nodes. This implies that if we apply coding only on the final output, faults during the first step cannot be recovered before the transpose step and the computation might halt before proceeding to the second step. In this work, we propose a coding strategy which can protect faults during each step so that the computation would not fail due to earlier errors.

In our coding strategy, we utilize $(P-K)$ redundant processors to encode
the first and the second FFT steps separately. In the first step, processors perform FFT on the row-wise data $X_i^{(row)}$'s. In order to protect from the loss of output at a failed node, we have to encode parity symbols across columns (column-wise encoding). By doing this, at the end of the first step, any successful $K$ processors can recover the output and proceed to the next step. In the second step, each processor computes FFT on the column-wise data, $Y_i^{(col)}$'s, so we encode row-wise parity symbols. Our coded computing algorithm is described below (*: additional steps that are not present in the uncoded algorithm).

\begin{algorithm}[Coded Distributed FFT Algorithm] \label{alg:coded_fft}
\mbox{}
\vspace*{-\parsep}
\vspace*{-\baselineskip}

\begin{enumerate}
\item \label{stp:enc1} * Encode column-wise parity symbols at each processor.
\begin{equation}
\tilde{X} = G_1^T X = \begin{bmatrix}
\tilde{X}^{(row)}_1 \\ 
\vdots \\
\tilde{X}^{(row)}_{P} 
\end{bmatrix}
\end{equation}
$G_1$ is an $N_1$-by-$N_1'$ encoding matrix for where $N_1'= \frac{P}{K}N_1$:
\begin{equation}
G_1 = \begin{bmatrix} 
I_{N_1} & \mathcal{P}_1 
\end{bmatrix}
\end{equation}

\item Rearrange the encoded data. Now the $i$-th processor has $\tilde{X}^{(row)}_{i}$.

\item Compute $N_1/K$ row-wise FFTs of size $N_2$ at each processor.
\item *Wait for the first successful $K$ processors.
\item Transpose the output within the successful $K$ processors.
\item \label{stp:dec1} * If needed, decode to retrieve the uncoded output at each processor.
\item Multiply twiddle factors.
\item \label{stp:enc2} * Encode row-wise parity symbols and send them to the remaining $P-K$ processors.
\begin{equation}
\tilde{Y} =  Y G_2 = \begin{bmatrix}
\tilde{Y}^{(col)}_1 & \cdots & \tilde{Y}^{(col)}_{P} 
\end{bmatrix}
\end{equation}
$G_2$ is an $N_2$-by-$N_2'$ encoding matrix where $N_2'= \frac{P}{K}N_2$:
\begin{equation}
G_2 = \begin{bmatrix} 
I_{N_2} & \mathcal{P}_2 
\end{bmatrix}
\end{equation}

\item Compute $N_2/K$ row-wise FFTs of size $N_1$ at each processor.

\item * Wait for the first successful $K$ processors and halt the remaining $P-K$ processors. 

\item \label{stp:dec2} Decode if needed.
\end{enumerate}
\end{algorithm}

For both encoding steps in Step~\ref{stp:enc1} and Step~\ref{stp:enc2}, we use \emph{a $(P, K)$ systematic MDS code}. In the following theorem, we show that using the proposed coded distributed FFT algorithm, any $K$ successful processors are enough to recover the computed outputs at Step \ref{stp:dec1} and Step \ref{stp:dec2} \footnote{Note that we do not have any fault recovery for twiddle multiplication step. However, computational complexity of twiddle factor multiplication is $O(N)$ compared to that of $O(N \log N)$. Hence, it is less probable to have faults during twiddle factor multiplication step}.

\begin{theorem}
In Algorithm~\ref{alg:coded_fft} where we compute distributed FFT of size $N$ using $P$ processors each of which can store and process $\frac{1}{K}$ fraction of the input ($P > K$), any successful $K$ processors can recover $Y$ and $Z$ at Step~\ref{stp:dec1} and \ref{stp:dec2}, respectively.
\end{theorem}
\begin{proof}
Let us first prove that we can $Y$ with any $K$ successful processors at Step \ref{stp:dec1} and the similar argument holds for recovering $Z$ at step~\ref{stp:dec2}. 

At Step 4, we will have the result from $K$ successful workers. Let us denote the indices of the successful $K$ workers as $\{i_1, \cdots, i_K \}$. We thus have successful output:
\begin{equation}
Y_{\textnormal{suc}} = \begin{bmatrix}
\tilde{X}_{i_1}^{\textnormal{(row)}} F_{N_2} \\ 
\tilde{X}_{i_2}^{\textnormal{(row)}} F_{N_2} \\
\vdots \\ 
\tilde{X}_{i_K}^{\textnormal{(row)}} F_{N_2}
\end{bmatrix} = \begin{bmatrix}
Y_{\textnormal{suc},i_1}^{\textnormal{(col)}} & \cdots Y_{\textnormal{suc},i_K}^{\textnormal{(col)}}
\end{bmatrix}.
\end{equation}
After transposing at Step 5, processors $i_1, \cdots, i_K$ will have column-wise output $Y_{\textnormal{suc},i_1}^{\textnormal{(col)}}, \cdots Y_{\textnormal{suc},i_K}^{\textnormal{(col)}}$. $Y_{\textnormal{suc},i}^{\textnormal{(col)}}$ can be written as:
\begin{equation}
Y_{\textnormal{suc},i}^{\textnormal{(col)}} = G_{1, suc}^T X F_{N_2, i}^{\textnormal{(col)}} =  G_{1, suc}^T Y_i^{\textnormal{(col)}}
\end{equation}
where $G_{1,suc}^T$ is a submatrix of $G_1^T$ which only has rows from successful nodes and hence has the size $N_1$-by-$N_1$.

As we assume the erasure model where we lose the entire data from a failed node, we only code across nodes, not within a node. Hence, our encoding matrix $G_1$ has the following structure:
\begin{equation}
G_1 = \mathcal{G}_{1} \otimes I_{N_1/K}
\end{equation}
where $\mathcal{G}_{1}$ is the encoding matrix for a systematic $(P, K)$-MDS code which has size $K$-by-$P$. 

Now, $G_{1, \textnormal{suc}}$ can be rewritten as:
\begin{equation}
G_{1, \textnormal{suc}}^T = \mathcal{G}_{1, \textnormal{suc}}^T \otimes I_{N_1/K}
\end{equation}
where $\mathcal{G}_{1, \textnormal{suc}}$ is a submatrix of $\mathcal{G}$ that only has $K$ columns from the $K$ successful nodes, i.e., $i_1$-th to $i_K$-th columns of $\mathcal{G}$.
Because $\mathcal{G}_{1}$ is a $(P, K)$ MDS code, $\mathcal{G}_{1, \textnormal{suc}}$ always has a full rank.
As $\textnormal{rank} (A \otimes B) = \textnormal{rank}(A) \cdot \textnormal{rank}(B)$ for any matrices $A$ and $B$, $\textnormal{rank}(G_{1, \textnormal{suc}}) = N_1$. Hence, we can recover $Y_i^{\textnormal{(col)}}$ at every successful node at Step \ref{stp:dec1}. Similar argument applies to recovering $Z$ at Step \ref{stp:dec2}.
\end{proof}

\section{Communication cost of encoding} \label{sec:comm_cost}
In this section, we analyze the communication cost of uncoded/coded distributed FFT algorithms, and show that as long as $P-K=o(\log K)$, the communication overhead of coding is negligible compared to the  communication cost of FFT.

Let us begin with understanding the communication cost of uncoded FFT algorithm. In Algorithm~\ref{alg:uncoded_fft}, only steps that require communication are Step \ref{stp:uc_rarrange} and \ref{stp:uc_transpose}. Both steps need communication to transpose the data stored in distributed processors. For transposing the data, all processors have to exchange data with all the other processors. This communication is known as \emph{``all-to-all''} communication. Bruck et al. showed lower bounds and explicit algorithms that achieve lower bounds for two special cases of  all-to-all communication \cite{bruck1997} -- a minimum-communication-rounds regime, and a minimum-bandwidth regime. Let us first formally define all-to-all communication.

\begin{definition}[All-to-all]
In all-to-all$(p, n)$ communication, there are $p$ nodes each of which stores $n$ symbols. The data stored in the $i$-th node can be broken down into $p$ data blocks, $M_{i,1}, \cdots M_{i,p}$, where the size of each block is $n/p$ symbols. The goal of the communication is to transpose the data stored in $p$ processors so that at the end of the communication, the $i$-th node has $M_{1,i}, \cdots, M_{p,i}$ data blocks.
\end{definition}

We will first give a simple lower bound of all-to-all$(p,n)$ communication. 

\begin{theorem}[Proposition 2.3 and 2.4 in \cite{bruck1997}]
For all-to-all$(p,n)$ communication, $C_1$ and $C_2$ are lower bounded by:
\begin{equation}
C_1 \ge \lceil \log_2 p \rceil, \quad C_2 \ge \frac{p-1}{p}n
\end{equation}
\end{theorem}

However, Bruck et al. showed that the lower bounds on $C_1$ and $C_2$ cannot be achieved simultaneously which is stated in the theorem below \cite{bruck1997}.

\begin{theorem} [Theorem 2.5 and 2.6 in \cite{bruck1997}] \label{thm:all_to_all}
If all-to-all$(p,n)$ communication uses the minimum number of rounds, i.e., $C_1 = \lceil \log_2 p \rceil$, $C_2$ is lower bounded by:
\begin{equation}
C_2 \ge \frac{n}{2} \log_2 p.
\end{equation}
If all-to-all$(p,n)$ communication uses the minimum number of symbols transferred in sequence, i.e., $C_2 = \frac{p-1}{p}n$ symbols in a sequence, then $C_1$ is lower bounded by: 
\begin{equation}
C_1 \ge p-1.
\end{equation}
Furthermore, both lower bounds are achievable.
\end{theorem}

Now, by using Theorem~\ref{thm:all_to_all}, we can give communication cost lower bounds on the transpose step in the distributed FFT algorithm.
\begin{corollary}
The transpose step of $N$-point FFT requires the communication cost at least
\begin{equation} \label{eq:transpose_min_rnd}
\lceil \log_2 K \rceil \alpha + \frac{N}{K} \log_2 K \beta
\end{equation}
when using the minimum communication rounds regime, and 
\begin{equation}
(K-1) \alpha + \frac{(K-1)}{K} \frac{N}{K} \beta
\end{equation}
when using the minimum communication bits regime.
\end{corollary}

Under our system model $N$ is very large, $\log N << \sqrt{N}$. Hence, we should always choose the minimum-communication-round regime over the minimum-bandwidth regime. From now on, \emph{we will only consider minimum communication round regime} and use its communication cost given in (\ref{eq:transpose_min_rnd}).

Now, let us identify additional communication cost due to coding in Algorithm~\ref{alg:coded_fft}. In the first encoding step where we compute column-wise parity symbols, we do not need any communication since processors already have column-wise data in the beginning. Also, for the first decoding in Step \ref{stp:dec1}, column-wise decoding can be done in local processors as each processor has column-wise data again after the transpose step. In Step~\ref{stp:enc2}, it requires inter-processor communication to encode row-wise parity symbols as one row of the data is spread over all the processors \footnote{For the communication cost of the last decoding step, see Remark~\ref{rem:last_dec}}. 
Hence, in this section, \emph{we will analyze the communication cost of the second encoding step where we compute}
\begin{equation}
\tilde{Y} = Y G_2.
\end{equation}

Before we begin encoding communication cost analysis, we want to make a few remarks.

\begin{remark} \label{rem:enc_adv}
If we can do the second encoding, which is computing row-wise parity symbols, at local processors before the transpose step, we can avoid communication for distributed encoding at Step \ref{stp:enc2}. However, there is no trivial way of doing this using a linear code due to the twiddle factors. After Step 3, the $i$-th processor has
\begin{equation}
Y_i^{\textnormal{(row)}} = \tilde{X}_i^{\textnormal{(row)}}F_{N_1} = G_{1, i}^{\textnormal{(row)}} X F_{N_1}.
\end{equation}
If we do row-wise encoding at the $i$-th processor locally before the transpose step, the $i$-th processor will have
\begin{equation}
\tilde{Y}_i^{\textnormal{(row)}} = G_{1, i}^{\textnormal{(row)}} X F_{N_2} G_2.
\end{equation}
We then perform the transpose of the output from the first $K$ successful nodes. The $i$-th node now has 
\begin{equation}
\tilde{Y}_i^{\textnormal{(col)}} = G_{1, suc} X F_{N_2} G_{2, i}^{\textnormal{col}}.
\end{equation}
Column-wise decoding can be done locally by inverting $G_{1, suc}$:
\begin{equation}
\hat{Y}_i^{\textnormal{(col)}} = G_{1, suc}^{-1} G_{1, suc} X F_{N_2} G_{2, i}^{\textnormal{col}} = X F_{N_2} G_{2, i}^{\textnormal{col}}.
\end{equation}
We now have to multiply twiddle factors to $\hat{Y}_i^{\textnormal{(col)}}$:
\begin{equation}
\hat{Y}_i^{\textnormal{(col)}} = T_N \circ \hat{Y}_i^{\textnormal{(col)}} = T_N \circ (X F_{N_2} G_{2, i}^{\textnormal{col}})
\end{equation}
However, this will produce a different final output from what we expect because of the nonlinearity of Hadamard product:
\begin{equation}
A \circ (BC) \neq (A \circ B) C.
\end{equation}
Hence,
\begin{equation}
T_{N,i}^{(col)} \circ (X F_{N_2} G_{2, i}^{\textnormal{col}}) \neq (T_{N,i}^{(col)} \circ X F_{N_2}) G_{2, i}^{\textnormal{col}}.
\end{equation}
From our modified coding strategy, our final output from successful nodes will be $F_{N_1} T_N \circ (X F_{N_2} G_{2,\textnormal{suc}})$ and even after decoding, we will have 
\begin{equation}
F_{N_1} T_N \circ (X F_{N_2} G_{2,\textnormal{suc}})G_{2,\textnormal{suc}}^{-1} \neq F_{N_1} T_N \circ (X F_{N_2} ).
\end{equation}
This means that we have to perform twiddle factor multiplication before proceeding to the row-wise encoding step. With the same argument, we can show that column-wise decoding must be done before multiplying twiddle factors. It concludes that because of the twiddle factors, the second-step encoding must be done across the processors incurring some communication cost.
\end{remark}

\begin{remark} \label{rem:last_dec}
The last decoding in Step~\ref{stp:dec2} can also incur communication overhead in our fully-distributed setting. However, we do not consider this here. It is because we assume that
the final output will be stored encoded at distributed nodes until it is needed. We can think of two use cases. First, a central server needs the raw output. In this case, all the nodes will transfer their encoded outputs to the central server and decoding can be done at the central server. The second case is where we want to do some processing on the frequency domain of the data. For instance, when we want apply a certain filter on the data, we need to do element-wise multiplication on the Fourier transformed data. In this case, we might want to decode the output in a distributedly before going into the next computation. It would be an interesting direction of future study to identify the communication cost of decoding and explore possibilities of applying distributed storage codes \cite{dimakis2011survey,rashmi2011optimal} in order to minimize decoding communication overhead.
\end{remark}

\begin{figure}[ht]
    \centering
    \begin{subfigure}[b]{0.35\textwidth}
        \includegraphics[width=\textwidth]{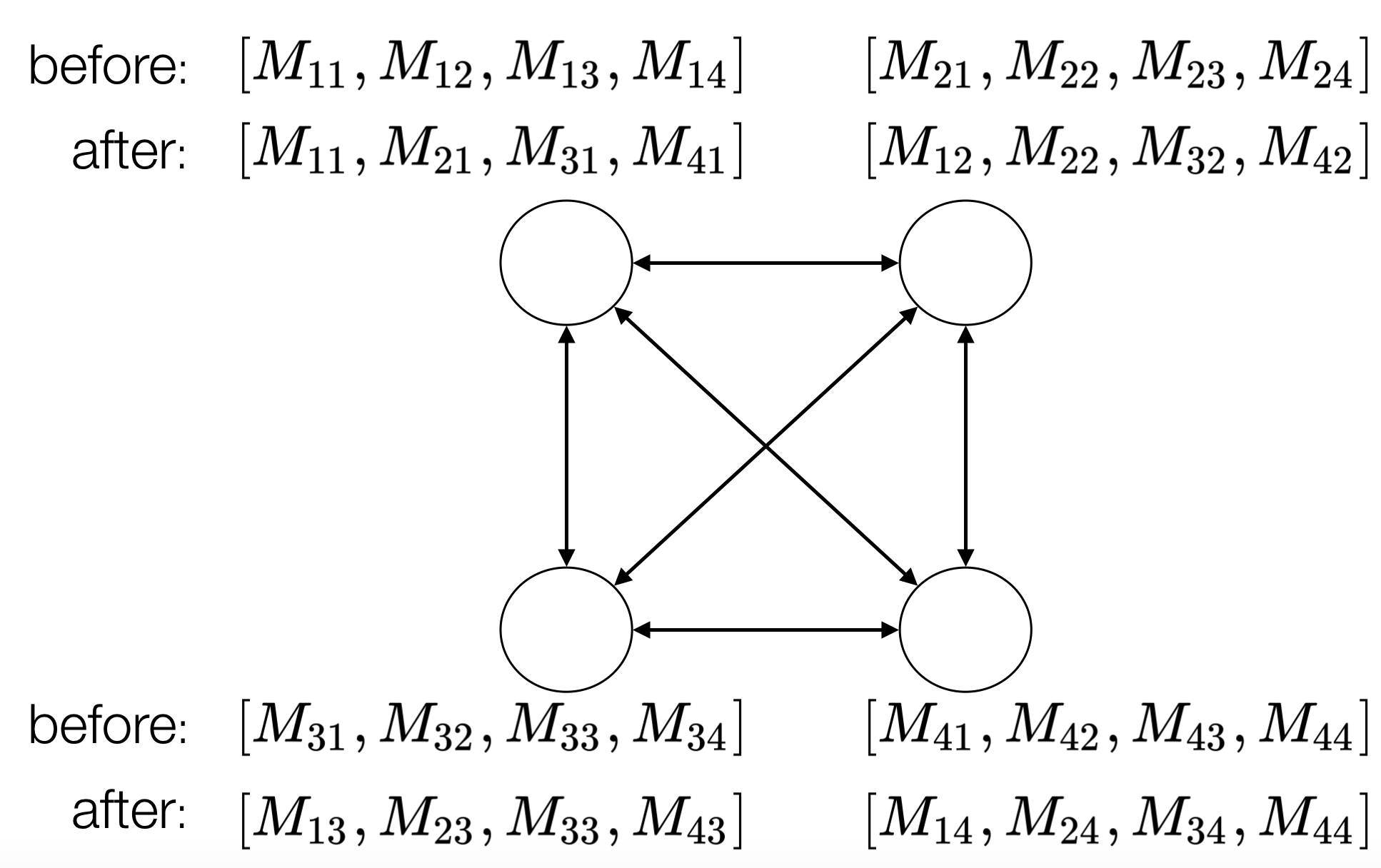}
        \caption{All-to-all$(4,n)$}
        \label{fig:all_to_all}
    \end{subfigure}
\quad
    \begin{subfigure}[b]{0.28\textwidth}
        \includegraphics[width=\textwidth]{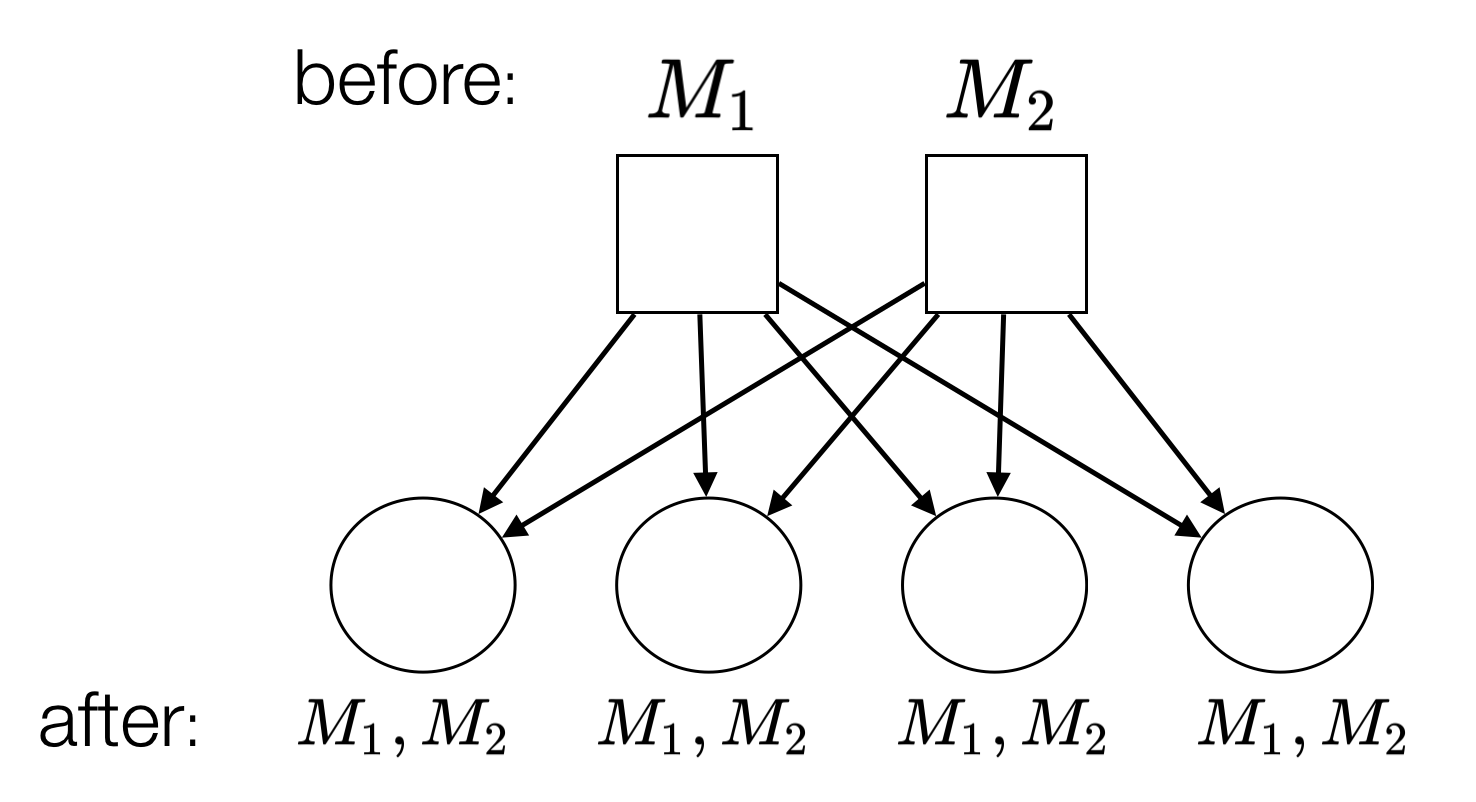}
        \caption{Multi-broadcast$(4,2,n)$}
        \label{fig:multi_broadcast}
    \end{subfigure}
\quad
    \begin{subfigure}[b]{0.3\textwidth}
        \includegraphics[width=\textwidth]{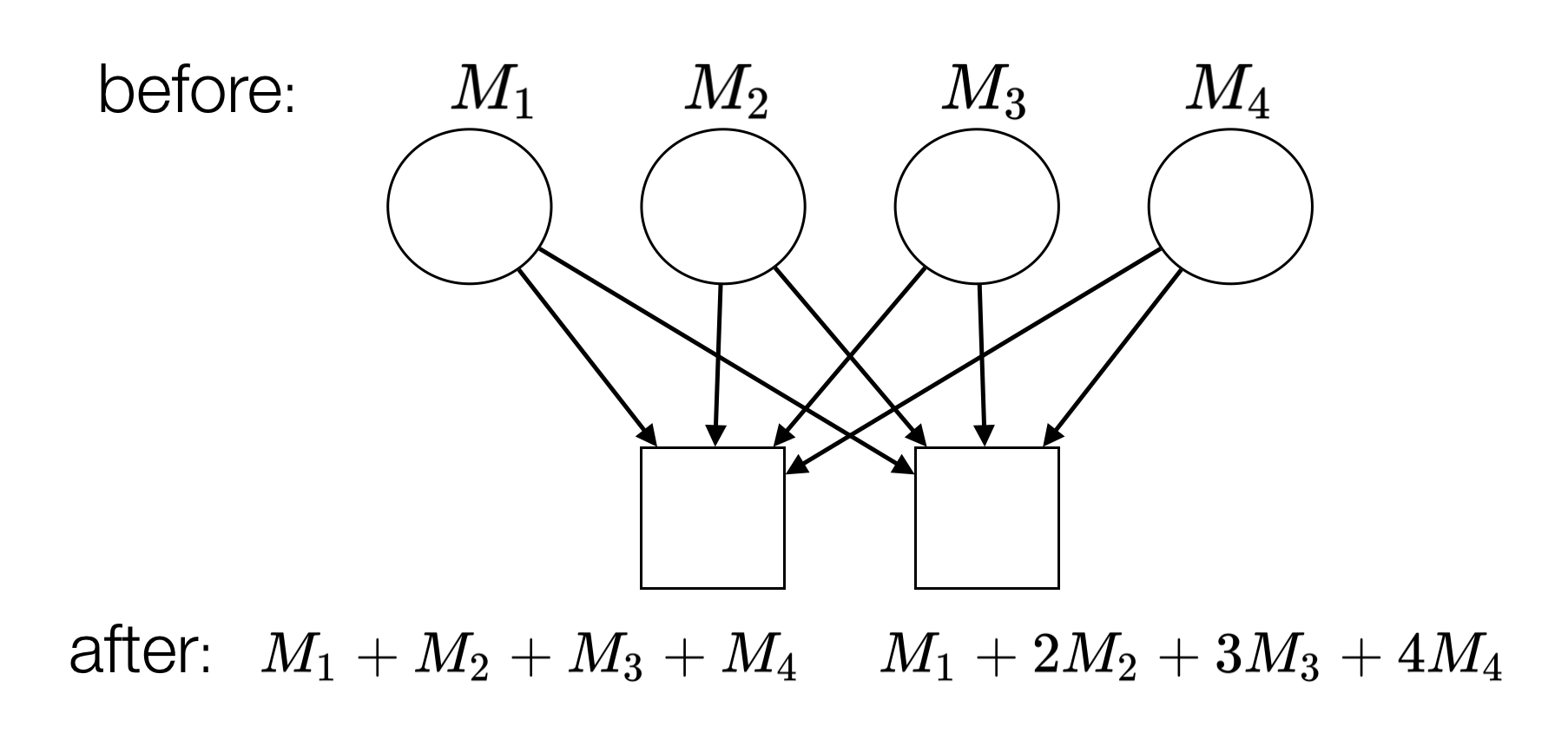}
        \caption{Multi-reduce$(4,2,n)$}
        \label{fig:multi_reduce}
    \end{subfigure}
    \caption{Example of all-to-all, multi-broadcast, multi-reduce communications. In all examples each node first has a message with $n$ symbols. In the all-to-all example, $M_{ij}$  has $n/4$ symbols for $i.j=1,\cdots,4$.}\label{fig:animals}
\end{figure}

We now want to analyze the communication cost of the second encoding step. Let us first investigate the communication cost of a simple encoding scheme where we add one parity node that stores the checksums of data, $X_1 +  \cdots + X_K$.  The encoding matrix $G_2$ for this can be written as follows:
\begin{equation}
\mathcal{G}_{\textnormal{cks}} = \small \begin{bmatrix} 
& & & 1\\
& I_{K} & & \vdots\\
& & & 1 \\
\end{bmatrix}
\end{equation}
\begin{equation}
G_2 = \mathcal{G}_{\textnormal{cks}} \otimes I_{N_2/K}
\end{equation}
For this computation, all $K$ nodes have to send its data to one checksum node to compute the sum of all the data in the network. This is a well-known communication operation called \emph{``reduce(-to-one)''}. 

\begin{definition}[Reduce]
In reduce$(p,n)$ communication, there are $p$ data nodes which have data $M_1, \cdots, M_p$ of size $n$ and one reduction node. The goal of the communication is to send $M_1+\cdots+M_p$ to the reduction node.
\end{definition}

A lower bound on the communication cost of reduce$(p,n)$ operation is given in the following theorem.

\begin{theorem}
The communication cost of reduce$(p,n)$ is lower bounded by
\begin{equation} \label{eq:rdc_lower}
\lceil \log_2 p \rceil \alpha + n \beta.
\end{equation}
\end{theorem}

It was found that reduce operation can be done by reversing any broadcasting algorithm, where one broadcasting node sends its message to all the other processors in the network. Traff and Ripke\cite{traff2008} proposed a near-optimal broadcasting algorithm that achieves the lower bound (\ref{eq:rdc_lower}) within a factor of $2$ . By reversing their broadcasting algorithm, we can achieve the same communication cost for reduce$(p,n)$ communication. 

\begin{theorem} \label{thm:rdc_upper}
Reduce$(p,n)$ can be done with the communication cost of at most
\begin{equation} \label{eq:rdc_upper}
(\sqrt{\lceil \log_2 p \rceil \alpha}  + \sqrt{n \beta} )^2 \le 2(\lceil \log_2 p \rceil \alpha + n \beta).
\end{equation}
\end{theorem}

Whether (\ref{eq:rdc_upper}) is optimal or not is an open problem. We will use this as a state-of-the-art communication algorithm for reduce operation. By applying (\ref{eq:rdc_upper}), we can obtain the communication cost for encoding one checksum node.

\begin{corollary}
A $(K+1, K, 2)$ systematic MDS code over $K$ systematic processors each of which hs $N/K$ data symbols can be encoded with the communication cost of 
\begin{equation}
(\sqrt{\lceil \log_2 K \rceil\alpha}  + \sqrt{N/K \beta} )^2 \le 2 (\lceil \log_2 K \rceil \alpha + N/K \beta).
\end{equation}
\end{corollary}

We can now extend computing checksums to computing parity symbols for a generic $(P,K,d=P-K+1)$ systematic MDS code. Unlike checksum computation which only requires a single reduce(-to-one) operation, here we need multiple reductions to $P-K$ nodes. 

From the intuition we got from reduce(-to-one) problem, we will first establish bounds for multi-broadcasting problem (will be defined below) and show that multi-reduce problem for encoding a $(P, K, d=P-K+1)$ systematic MDS code can be solved by reversing the multi-broadcasting algorithm. 

\begin{definition} [Multi-broadcast]
In  multi-broadcast$(p, r, n)$ communication, there are $r$ broadcasting nodes and $p$ destination nodes. Broadcasting nodes have distinct messages $M_1, \cdots, M_r$ of size $n$ symbols. At the end of the communication, all $p$ destination nodes should have all $r$ messages, $M_1, \cdots, M_r$. 
\end{definition}

We want to note that multi-message broadcasting has been studied in the literature \cite{sanders2003,barnoy1994}. However, their models have one broadcasting node which sends multiple messages in a sequence. This is fundamentally different from our \textit{multi-broadcast} which has multiple broadcasting nodes that can send out their messages simultaneously. To the best of our knowledge, communication cost analysis of this specific problem has not been studied before.

We will first show a communication algorithm for multi-broadcast$(p,r,n)$ and then show that it achieves the lower bound within a factor of $2$.

\begin{theorem} \label{thm:multi_bcast_up}
Multi-broadcast$(p, r, n)$ can be done with the communication cost at most
\begin{equation} \label{eq:multi_bcast_up}
2(\lceil \log_2 p \rceil \alpha + rn \beta)
\end{equation}
\end{theorem}

\begin{proof}
First, divide $p$ processors into $r$ disjoint sets of size $p/r$. Let us denote the sets as $S_1, S_2, \cdots, S_r$. The $i$-th broadcasting node broadcasts its message to all the nodes in $S_i$. With the optimal broadcasting algorithm \cite{traff2008}, it takes communication cost of $(\sqrt{\log_2 \frac{p}{r} \alpha}  + \sqrt{n \beta} )^2$.

After the broadcasting step, the $j$-th nodes in $S_i$'s $(i=1,\cdots, r)$ communicate with each other so that all of them can share $M_1, \cdots, M_r$. This is all-gather$(r,n)$ communication which is defined as follows.
\begin{definition}[All-gather]
In all-gather$(p,n)$ communication, there are $p$ nodes which have distinct messages $M_1, \cdots, M_p$ of size $n$ symbols. At the end of the communication, all $p$ nodes should have all $p$ messages.
\end{definition}
All-gather$(r,n)$ can be done  with communication cost of $ ( \log_2 r ) \alpha  + (r-1) n \beta$ using the \emph{bidirectional algorithm}~\cite{chan2007collective}. 

The total communication cost of this two-step algorithm is
\begin{align*}
 (\sqrt{\log_2 \frac{p}{r} \alpha}  + \sqrt{n \beta} )^2 + \log_2 r \alpha  + (r-1) n \beta &\le 
\lceil \log_2 p \rceil \alpha + rn \beta + (\log_2 \frac{p}{r} \alpha + n \beta)  \label{eq:multi_bcast_up} \\
& \le 2 (\lceil \log_2 p \rceil \alpha + rn \beta).
\end{align*}
\end{proof}

We now show a lower bound for multi-broadcast$(p, r, n)$ communication.

\begin{theorem}\label{thm:multi_bcst_lower}
The communication cost of multi-broadcast$(p, r, n)$ is lower bounded by 
\begin{equation} \label{eq:multi_bcast_low}
\lceil \log_2 p \rceil \alpha + rn \beta 
\end{equation}
\end{theorem}
\begin{proof}
Each broadcasting node must communicate to $p$ destination nodes which takes at least $\lceil \log_2 p \rceil$ communication rounds. Each destination node has to receive messages $M_1, \cdots M_r$ which have  $n$. Hence, mutlti-broadcast$(p, r, n)$ requires at least the bandwidth of $rn$. 
\end{proof}

By comparing (\ref{eq:multi_bcast_up}) and (\ref{eq:multi_bcast_low}), we can see that the algorithm given in Theorem~\ref{thm:multi_bcast_up} achieves the lower bound within a factor of $2$. 

Finally, we define \emph{multi-reduce} operation which is the communication required for encoding parity symbols, and show that it can be done with the same communication cost as multi-broadcast operation. 

\begin{definition} [Multi-reduce]
In multi-broadcast$(p, r, n)$ communication, there are $p$ data nodes and $r$ reduction nodes ($r < p$). $p$ data nodes have data $M_1, \cdots, M_p$ each of which consist of $n$ symbols. At the end of communication, the $i$-th reduction node will have $a_{i,1} M_1 + \cdots + a_{i,p} M_p$ where $a_{i,j}$'s ($i = 1, \cdots, r, j= 1, \cdots, p$) are chosen so that the data from any $p$ nodes are linearly independent combinations of $M_1, \cdots, M_p$. 
\end{definition}

\begin{theorem} \label{thm:multi_rdc_bcast}
Multi-reduce$(p,r,n)$ communication can be done by reversing the multi-broadcast algorithm given in Theorem~\ref{thm:multi_bcast_up}. Hence, the communication cost of multi-reduce$(p,r,n)$ is at most 
\begin{equation}
2(\lceil \log_2 p \rceil \alpha + rn \beta )
\end{equation}
\end{theorem}
\begin{proof}
Let $D_1, D_2, \cdots, D_p$ denote the data at $p$ data processors. Let us divide data processors into $r$ disjoint sets of size $p/r$ and let $S_i$ denote the set of indices of the $i$-th set: $S_i = \{ (i-1) \cdot p/r+1, \cdots, (i-1)\cdot p/r + p/r \}$. This is all-gather$(r,n)$ communication.

First, the $j$-th nodes in $S_i$'s $(i= 1, \cdots, n)$ perform all-gather communication. All the $j$-th processors in $S_i$'s will have $D_{j}, D_{j+p/r}, \cdots , D_{j+(r-1)p/r}$ after the communication. 

In the second step, all the nodes in $S_i$ will carry out reduce communication with the $i$-th reduction node. Each node in $S_i$ will compute a corresponding linear combination of the the data it has and send only $n$ symbols of data to the $i$-th reduction node. For instance, the $j$-th node in $S_i$ will compute
\begin{equation*}
a_{i,j} D_j + a_{i,j+p/r} D_{j+p/r} + \cdots + a_{i, j+(r-1)p/r} D_{j+(r-1)p/r}.
\end{equation*}
This is reduce$(p/r, n)$ which can be done with the communication cost of $(\sqrt{\log_2\frac{p}{r} \alpha} + \sqrt{n\beta})^2$. This completes multi-reduce$(p,r,n)$ communication.
\end{proof}

This gives an achievable communication scheme for encoding parity symbols of a $(P, K, d)$ systematic MDS code.

\begin{corollary}
A $(P, K, d=P-K+1)$ systematic MDS code over $K$ systematic processors each of which has $N/K$ data symbols can be encoded with the communication cost of 
\begin{equation} \label{eq:enc_comm_up}
2 (\lceil \log_2 K \rceil \alpha + (P-K) \frac{N}{K} \beta).
\end{equation}
\end{corollary}

By comparing the encoding communication overhead given in (\ref{eq:enc_comm_up}) with the communication cost of uncoded FFT algorithm given in (\ref{eq:transpose_min_rnd}), we can prove our main theorem.

\begin{theorem}\label{thm:main_result}
In our proposed coded FFT algorithm, if $P-K = o(\log_2 K)$, communication overhead of coding is negligible compared to the communication cost of uncoded FFT. 
\end{theorem}
\begin{proof}
Uncoded FFT algorithm requires two transpose operation, one in the beginning and one before the second FFT step. This requires communication cost of 
\begin{equation}
2 \left( \lceil \log_2 K \rceil \alpha +  \frac{N}{2K} \lceil \log_2 K \rceil \beta \right)
\end{equation}
If we compare this against the communication cost of encoding given in 
(\ref{eq:enc_comm_up}), the condition for the encoding cost to be smaller than the all-to-all communication is given as follows: 
\begin{align}\label{eq:enc_transpose1}
2 \left(\lceil \log_2 K \rceil \alpha + (P-K)\frac{N}{K} \beta \right)   \nonumber  &< 2 \Big( \lceil \log_2 K \rceil \alpha +  \frac{N}{2K} \lceil \log_2 K \rceil \beta \Big) \\
P-K &< \frac{\log_2 K}{2}.
\end{align}
Hence, as long as $P-K$ is smaller than $\frac{\log_2 K}{2}$ in scaling sense, communication overhead of coding is negligible compared to the intrinsic communication cost of uncoded distributed FFT algorithm.
\end{proof}

\section{Future Work}
Extending communication cost analysis to non-MDS codes can be a natural direction of  future study. Especially, codes that have sparse generator matrices such as LT codes~\cite{luby2002lt,severinson2017block} might be able to reduce communication overhead substantially while sacrificing error correction capability. Establishing bounds on the trade-off between encoding communication cost and error correction capability would be interesting.

\bibliographystyle{IEEEtran}
\bibliography{IEEEabrv,reference}
\end{document}